\newif\iflong\longtrue

\documentclass{article}

\usepackage{arxiv}
\usepackage[utf8]{inputenc}
\usepackage[T1]{fontenc}
\usepackage{amssymb}
\usepackage[arrow, matrix, curve]{xy}
\usepackage[]{todonotes}
\usepackage{graphicx}
\usepackage{float}
\usepackage{dsfont}
\usepackage{braket}
\usepackage{tikz}
\usetikzlibrary{arrows,automata}
\usepackage{verbatim}
\usepackage{mathrsfs}
\usepackage{xspace}
\usepackage{mathtools}
\usepackage{nicefrac}
\usepackage{apxproof}
\usepackage[noend]{algpseudocode}
\usepackage{booktabs}
\usepackage{skull}
\usepackage{apxproof}
\usepackage{cleveref}

\usepackage{amsthm}
\theoremstyle{definition}

\theoremstyle{plain}
\newtheorem{theorem}{Theorem}

\theoremstyle{remark}

\newtheorem{lemma}{Lemma}

\newif\iflncs\lncsfalse
\newcommand{\lncsqed}{\iflncs\hfill$\qed$\fi}

\newcommand{\Oh}{\ensuremath{\mathcal{O}}}

\newcommand{\LCM}{\ensuremath{\operatorname{lcm}}\xspace}
\newcommand{\Atr}{\ensuremath{\mathit{Attr}}\xspace}

\newcommand{\todomi}[2][]{\todo[inline,#1,color=green!50]{#2}}

\definecolor{myred}{rgb}{1,0.25,0.25}

\newtheorem{myclaim}{Claim}{\itshape}{\rmfamily}
\newenvironment{claimproof}{{\noindent\textit{Proof. }}}{\hfill$\blacksquare$}

\newcommand{\prob}[3]{\begin{quote}  \textsc{#1}\\  \textbf{Input:} #2\\  \textbf{Question:} #3\end{quote}}

\newcommand{\W}[1]{\ensuremath{\mathrm{W}[#1]}\xspace}
\newcommand\NP{\ensuremath{\mathrm{NP}}\xspace}
\newcommand\PSPACE{\ensuremath{\mathrm{PSPACE}}\xspace}

\newcommand\PTIME{\ensuremath{\mathrm{PTIME}}\xspace}
\newcommand\EXPTIME{\ensuremath{\mathrm{EXPTIME}}\xspace}

\newcommand{\modid}[2]{\ensuremath{{#1}\smash{\left[{#2}\right]^{\circ}}}}

\newcommand{\CNR}{\textsc{Periodic Cop \& Robber}\xspace}
\newcommand{\PCA}{\textsc{Periodic Character Alignment}\xspace}

\newcommand{\att}{\ensuremath{\operatorname{Attr}}}

	\title{A Timecop's Chase Around the Table}
	\author{Nils Morawietz\thanks{Supported by Deutsche Forschungsgemeinschaft, project OPERAH, KO~3669/5-1.}\ \ and Petra Wolf\thanks{Supported by Deutsche Forschungsgemeinschaft, project	FE 560/9-1.} \\
		$^*$Philipps-Universität Marburg, Fachbereich Mathematik und Informatik,  Marburg, Germany\\ {morawietz@informatik.uni-marburg.de} \and
	$^\dagger$Universität Trier, Fachbereich 4, Informatikwissenschaften, Trier, Germany\\
	{wolfp@informatik.uni-trier.de}}
\begin{document}
	%
	%
	%
	%
	%
	
\maketitle  
	\begin{abstract}
		We consider the cops and robbers game variant consisting of one cop and one robber on time-varying graphs (TVG). The considered TVGs are edge periodic graphs, i.e., for each edge, a binary string $s_e$ determines in which time step the edge is present, namely the edge $e$ is present in time step $t$ if and only if the string $s_e$ contains a $1$ at position $t \mod |s_e|$. This periodicity allows for a compact representation of an infinite TVG. 
		We prove that even for very simple underlying graphs, i.e., directed and undirected cycles the problem whether a cop-winning strategy exists is \NP-hard and \W1-hard parameterized by the number of vertices. 
		Our second main result are matching lower bounds for the ratio between the length of the underlying cycle and the least common multiple (\LCM) of the lengths of binary strings describing edge-periodicies over which the graph is robber-winning. Our third main result improves the previously known \EXPTIME upper bound for \CNR on general edge periodic graphs to \PSPACE-membership.

	\end{abstract}

\section{Introduction}
In general, a \emph{time-varying graph} (TVG) describes a graph that varies over time. For most applications, this variation is limited to the availability or weight of edges meaning that edges are only present at certain time-steps or the time needed to cross an edge changes over time. 
TVGs are of great interest in the area of \emph{dynamic networks}~\cite{DBLP:journals/paapp/CasteigtsFQS12,ganguly2009dynamics,holme2015modern,holme2012temporal} such as \emph{mobile ad hoc networks}~\cite{DBLP:journals/comsur/Zhang06} and \emph{vehicular networks} modeling traffic load factors on a road network~\cite{DBLP:conf/edbt/DingYQ08}. 
In those networks, the topology naturally changes over time and TVGs are used to reflect this dynamic behavior. 
Quite recently, TVGs became of interest in the context of graph games such as competitive diffusion games and Voronoi games~\cite{BF+21}.
There are plenty of representations for TVGs in the literature which are not equivalent in general. 
For instance, in~\cite{DBLP:journals/paapp/CasteigtsFQS12} a TVG is defined as a tuple $\mathcal{G} = (V, E, \mathcal{T}, \rho, \zeta)$ where $V$ is a set of vertices, $E \subseteq V \times V \times L$ is a set of labeled edges (with labels from a set $L$), $\mathcal{T} \subseteq \mathbb{T}$ is the \emph{lifetime} of the graph, $\mathbb{T}$ is the temporal domain and assumed to be $\mathbb{N}$ for discrete systems and $\mathbb{R}^+$ for continuous systems, $\rho \colon E \times \mathcal{T} \to \{0,1\}$ is the \emph{presence function} indicating whether an edge $e$ is present in time step $t$, and $\zeta \colon E \times \mathcal{T} \to \mathbb{T}$ is the latency function indicating the time needed to cross edge $e$ in time step $t$. We call the graph $G = (V, E)$ the \emph{underlying graph} of $\mathcal{G}$.
The literature has not yet agreed on how the function $\rho$ (and $\zeta$) are given in the input. This is of significant importance, if $\rho$ exhibits periodicy with respect to single edges in the context of computational complexity.
We say that a TVG belongs to the class of TVGs featuring \emph{periodicity of edges}, defined as class~8 in~\cite{DBLP:journals/paapp/CasteigtsFQS12}, if $\forall e\in E, \forall t \in \mathcal{T}, \forall k \in \mathbb{N}, \rho(e, t) = \rho(e, t+ kp_e)$ for some $p_e \in \mathbb{T}$, depending on $e$, and the underlying graph $G$ is connected. 
For these TVGs, the function $\rho$ can be represented for each edge $e\in E$ as a binary string of size $p_e$ concatenating the values of $\rho(e, t)$ for $0 \leq t <p_e$. 
Note that the period of the whole graph $\mathcal{G}$ is then the least common multiple (\LCM for short) of all string lengths $p_e$ describing edge periods. Therefore, the underlying graph ${G}$ of $\mathcal{G}$ can have exponentially many different sub-graphs $G_t$ representing a snapshot of $\mathcal{G}$ at time $t$.
This exponential blow-up is a huge challenge in determining the precise complexity of problems for TVGs featuring periodicity of edges as discussed in more detail in Section~\ref{sec:upperbounds} and \ref{sec:discussion}. 
Often, for general TVGs a representation containing all sub-graphs representing snapshots over the whole lifetime of the graph is chosen when the complexity of decision problems over TVGs are considered~\cite{bhadra2003complexity,DBLP:journals/tcs/MichailS16}.
An approach to unify the representation of TVGs is given in~\cite{wehmuth2015unifying}, also including the existence of vertices being effected over time. 
This approach represents $\rho(e, t)$ by enhancing an edge $e = (u, v)$ with the departure time $t_d$ at $u$ and the arrival time $t_a$ at $v$, where $t_a$ might be smaller than $t_d$ in order to model periodicity. 
As for TVGs with periodicity of edges where $\rho$ is represented as a binary string for each edge the periodicity of the TVG $\mathcal{G}$ might be exponential in its representation using the approach in~\cite{wehmuth2015unifying} would also cause an exponential blow-up in the 
representation of $\mathcal{G}$, as a decrement of the time value could only be used after a whole period of the graph, rather than after the period of one edge.
An other class of TVGs based on periodicity was considered in the field of robotics to model motion planning tasks if time dependent obstacles are present~\cite{sutner1988motion}. There, the availability of the vertices in the graph changes periodically and each edge needs a constant number of time steps to be crossed. An edge $e = \{u, v\}$ is only present if in the time span needed to cross $e$ both endpoints are $u$ and $v$ are still present.
In~\cite{sutner1988motion} the periodic function describing the availability of a vertex and the function describing the time needed to cross an edge is represented by an on-line program and can hence handle values exponential in their representation. 
This is crucial in the \PSPACE-hardness proof of the reachability problem for graphs in this class presented in~\cite{sutner1988motion}. There, the hardness is obtained by a reduction from the halting problem for linear space-bounded deterministic Turing machines where a configuration of the Turing machine is encoded in the time step. In the reduction, the periodicity of a single vertex as well as the time needed to cross an edge is of value exponential in the tape length-bound. Note that this representation of periodicity is exponentially more compact than in our setting and thus the result of~\cite{sutner1988motion} does not translate to our setting.

We will stick in the following to the model describing TVG featuring periodicity of edges where the function $\rho(e, t)$ is represented as a binary string. For this representation, Erlebach and Spooner~\cite{erlebach2020game} introduced recently a variant of the famous cops and robber problem which is intensively studied for static graphs~\cite{bonato2011game}. 
In the static setting, the game is played on a given graph and includes a single robber and a set of $k$-cops. The cops and robber occupy vertices and the cops choose their vertices first. Then, in each round the players alternate turns and the cops move first. Thereby, each cop can move to an adjacent vertex or pass and stay on her vertex. The same holds for the robber. We say that a graph is \emph{$k$-cop-winning}, if there exists a strategy for the $k$ cops in which they finally catch the robber, i.e., a cop occupies the same vertex as the robber. 
If the context is clear, we call a 1-cop-winning graph a \emph{cop-winning} graph.
If a graph is not cop-winning, we call it \emph{robber-winning}. 
A special interest on the game of cops and robbers lied in characterizing graphs which are $k$-cop-winning. While for one cop, the cop-winning graphs where understood in 1978 and independently 1983~\cite{quilliot1978jeux,DBLP:journals/dm/NowakowskiW83} as those featuring a special kind of ordering on the vertex set, called a \emph{cop-win} or \emph{eliminating} order, the case for $k$ cops was long open and solved in 2009 by exploiting a linear structure of a certain power of the graph~\cite{DBLP:journals/dm/ClarkeM12}.

In 2020 Erlebach and Spooner~\cite{erlebach2020game} connected the two discussed graph-theoretical topics of great interest by introducing a cops and robber game for \emph{edge periodic graphs}. These are TVGs featuring periodicity of edges as discussed above. They gave an algorithm which determines if the given edge periodic graph is $k$-cop-winning, which runs in time $\mathcal{O}(\LCM \cdot k\cdot |G|^{k+2})$, where $\LCM$ is the least common multiple of all length of binary strings describing $\rho(e, t)$ and $G$ is the underlying graph. As $\LCM$ can be exponentially in the input size, they proposed the question of whether this problem is \NP-hard. This question was answered positively for a one-cop version, also in 2020, by Morawietz, Rehs, and Weller~\cite{DBLP:conf/mfcs/MorawietzRW20} for TVGs of which the underlying graph has a constant size vertex cover or where two edges have to be removed to obtain a cycle. Moreover, they showed that the problem is~\W1-hard when parameterized by 
the size of the underlying graph $G$ even in these restricted cases, that is, they showed that there is presumably no algorithm solving the problem in time~$f(|G|)\cdot n^{\Oh(1)}$ for any computable function~$f$.
In other words, the exponential growth of the running time of every algorithm solving the problem has to  depend on  the length of binary strings describing $\rho(e, t)$.

\paragraph{Our contribution.} In this work, we show, that the \NP-hardness holds for even simpler classes of edge periodic graphs, namely for directed and undirected cycles. 
Moreover, we show that the~\W1-hardness when parameterized by $|G|$ even holds for these restricted instances (Section~\ref{sec:hardness}).
The quite restricted class of undirected cycles was also studied in~\cite{erlebach2020game} where an upper bound on the length of the cycle with respect to the \LCM was given for which each edge periodic cycle is robber winning. To be more precise, for an edge periodic cycle on $n$ vertices it holds that if $n \geq 2\cdot \ell \cdot \LCM$, then the graph is robber winning. Here, $\ell= 1$ if $\LCM$ is at least two times the longest size of a binary string describing $\rho(e, t)$ and $\ell=2$, otherwise. For these upper bounds only non-matching lower bounds where given, i.e., a family of cop-winning edge periodic cycles with length $3\cdot \LCM$ for $\ell= 2$ and $1.5 \cdot \LCM$ for $\ell=1$ are given. These lower bounds left a gap of size 
$0.5\cdot \ell \cdot \LCM$.
In this work, we show that the given upper bounds are tight by closing this gap by giving families of cop-winning edge periodic cycles of length $2 \cdot \ell \cdot \LCM -1$ (Section~\ref{sec:lowerBounds}).
Finally, we improve the currently best \EXPTIME upper bound shown in~\cite{erlebach2020game} for the problem, whether a given edge periodic graph is cop-winning to \PSPACE (Section~\ref{sec:upperbounds}). We conclude with a discussion on the restricted class of directed edge periodic cycles indicating that due to the compact representation of the edge periodic graphs, which does not introduce additional amounts of freedom, the standard complexity classes, such as \NP and \PSPACE might not be suitable to precisely characterize the complexity of this problem (Section~\ref{sec:discussion}).


\section{Preliminaries}
For a string $w= w_0w_1\dots w_n$ with $w_i \in \{0, 1\}$, for $0 \leq i \leq n$, we denote with $w[i]$ the symbol $w_i$ at position $i$ in $w$.
We write the concatenation of strings~$u$ and~$v$ as~$u \cdot v$. 
For non-negative integers $i \leq j$ we denote with $[i, j]$ the interval of natural numbers $n$ with $i \leq n \leq j$.

An \emph{edge periodic (temporal) graph} $\mathcal{G}=(V,E,\tau)$ (see also \cite{erlebach2020game})
consists of a graph $G=(V,E)$ (called the \emph{underlying graph})
and a function~$\tau:E \to \{0,1\}^*$ where $\tau$ maps each edge~$e$ to a sequence~$\tau(e)\in\{0,1\}^*$
such that $e$~exists in a time step~$t\geq 0$ if and only if $\modid{\tau(e)}{t} = 1$, where~$\modid{\tau(e)}{t} := \tau(e)[t \mod |\tau(e)|]$.
For an edge $e$ and non-negative integers $i \leq j$ we inductively define  $\modid{\tau(e)}{[i,j]}=\modid{\tau(e)}{i} \cdot \modid{\tau(e)}{[i+1,j]}$ and $\modid{\tau(e)}{[j,j]} = \modid{\tau(e)}{j}$.
If~$\tau(e) = 1$, we call~$e$ a~\emph{$1$-edge}.
Every edge~$e$ exists in at least one time step, that is, for each edge~$e$ there is some $t_e\in [0, |\tau(e)|-1]$ with $\tau(e)[t_e] = 1$. 
We might abbreviate $i$ repetitions of the same symbol $\sigma$ in $\tau(e)$ as $\sigma^i$.
Let $L_{\mathcal{G}} = \{|\tau(e)| \mid e \in E\}$ be the set of all edge periods of some edge periodic graph $\mathcal{G}=(V, E, \tau)$ and let $\LCM(L_{\mathcal{G}})$ be the least common multiple of all periods in $L_{\mathcal{G}}$.
We call an edge periodic graph $\mathcal{G}$ with an underlying graph consisting of a single cycle an \emph{edge periodic cycle}.
We denote with $\mathcal{G}(t)$ the sub-graph of $G$ present in time step $t$. We do not assume that $\mathcal{G}$ is connected in any time step. We will discuss directed and undirected edge periodic graphs. If not stated otherwise, we assume an edge periodic graph to be undirected. We illustrate the notion of edge periodic cycles in Figure~\ref{fig:example} showing an edge periodic cycle $\mathcal{G}$ together with $\mathcal{G}(t)$ for the first 5 time steps.
We now define the considered cops and robbers variant on edge periodic graphs with one single cop. Here, first the cop chooses her start vertex in $\mathcal{G}(0)$, then the robber chooses his start vertex in $\mathcal{G}(0)$. Then, in each time step $t$, the cop and robber move to an adjacent vertex over an edge which is present in $\mathcal{G}(t)$ or pass and stay on their vertex. 
 Thereby, the cop moves first, followed by the robber. We say that the cop catches the robber, if there is some time step $t$ for which the cop and the robber are on the same vertex after the cop moved. If the cop has a strategy to catch the robber regardless which start vertex the robber chooses, we say that $\mathcal{G}$ is \emph{cop-winning} and call the strategy implemented by the cop a \emph{cop-winning strategy}. If for all cop start vertices, there exists a start vertex and strategy for the robber to elude the cop indefinitely, we call $\mathcal{G}$ \emph{robber-winning}. The described game is a zero-sum game, i.e., $\mathcal{G}$ is either cop-winning or robber-winning.
\prob{\CNR}{An edge periodic graph $\mathcal{G}=(V,E,\tau)$.}
{Is $\mathcal{G}$ cop-winning?}
\begin{figure}[t]
	\centering
	\begin{tikzpicture}[node distance=20pt]
		\node[state, fill=black,minimum size=1pt] (0) at (0,0) {};
		\node[state, fill=black,minimum size=1pt, right of=0] (1){};
		\node[state, fill=black,minimum size=1pt, below of=1] (2){};
		\node[state, fill=black,minimum size=1pt, below of=2] (3){};
		\node[state, fill=black,minimum size=1pt, left of=3] (4){};
		\node[state, fill=black,minimum size=1pt, above of=4] (5){};
		\node[color=white] at (0.35,-2) {t};
		\path 
		(0) edge node[above] {01} (1) 
		(1) edge node[right] {1} (2)
		(2) edge node[right] {1011} (3)
		(3) edge node[below] {010} (4)
		(4) edge node[left] {10} (5)
		(5) edge node[left] {1110} (0);
	\end{tikzpicture}
\hspace{0.4cm}
	\begin{tikzpicture}[node distance=20pt]
	\node[state, fill=black,minimum size=1pt] (0) at (0,0) {};
	\node[state, fill=black,minimum size=1pt, right of=0] (1){};
	\node[state, fill=black,minimum size=1pt, below of=1] (2){};
	\node[state, fill=black,minimum size=1pt, below of=2] (3){};
	\node[state, fill=black,minimum size=1pt, left of=3] (4){};
	\node[state, fill=black,minimum size=1pt, above of=4] (5){};
	\node at (0.35,-2) {$t = 0$};
	\path 
	(1) edge (2)
	(2) edge (3)
	(4) edge (5)
	(5) edge (0);
\end{tikzpicture}
\hspace{0.4cm}
\begin{tikzpicture}[node distance=20pt]
	\node[state, fill=black,minimum size=1pt] (0) at (0,0) {};
	\node[state, fill=black,minimum size=1pt, right of=0] (1){};
	\node[state, fill=black,minimum size=1pt, below of=1] (2){};
	\node[state, fill=black,minimum size=1pt, below of=2] (3){};
	\node[state, fill=black,minimum size=1pt, left of=3] (4){};
	\node[state, fill=black,minimum size=1pt, above of=4] (5){};
	\node at (0.35,-2) {$t = 1$};
	\path 
	(0) edge (1) 
	(1) edge (2)
	(3) edge (4)
	(5) edge (0);
\end{tikzpicture}	
\hspace{0.4cm}
\begin{tikzpicture}[node distance=20pt]
\node[state, fill=black,minimum size=1pt] (0) at (0,0) {};
\node[state, fill=black,minimum size=1pt, right of=0] (1){};
\node[state, fill=black,minimum size=1pt, below of=1] (2){};
\node[state, fill=black,minimum size=1pt, below of=2] (3){};
\node[state, fill=black,minimum size=1pt, left of=3] (4){};
\node[state, fill=black,minimum size=1pt, above of=4] (5){};
\node at (0.35,-2) {$t = 2$};
\path 
(1) edge (2)
(2) edge (3)
(4) edge (5)
(5) edge (0);
\end{tikzpicture}
\hspace{0.4cm}
\begin{tikzpicture}[node distance=20pt]
	\node[state, fill=black,minimum size=1pt] (0) at (0,0) {};
	\node[state, fill=black,minimum size=1pt, right of=0] (1){};
	\node[state, fill=black,minimum size=1pt, below of=1] (2){};
	\node[state, fill=black,minimum size=1pt, below of=2] (3){};
	\node[state, fill=black,minimum size=1pt, left of=3] (4){};
	\node[state, fill=black,minimum size=1pt, above of=4] (5){};
	\node at (0.35,-2) {$t = 3$};
	\path 
	(0) edge (1) 
	(1) edge (2)
	(2) edge (3)
;
\end{tikzpicture}
\hspace{0.4cm}
\begin{tikzpicture}[node distance=20pt]
	\node[state, fill=black,minimum size=1pt] (0) at (0,0) {};
	\node[state, fill=black,minimum size=1pt, right of=0] (1){};
	\node[state, fill=black,minimum size=1pt, below of=1] (2){};
	\node[state, fill=black,minimum size=1pt, below of=2] (3){};
	\node[state, fill=black,minimum size=1pt, left of=3] (4){};
	\node[state, fill=black,minimum size=1pt, above of=4] (5){};
	\node at (0.35,-2) {$t = 4$};
	\path 
	(1) edge (2)
	(2) edge (3)
	(3) edge (4)
	(4) edge (5)
	(5) edge (0);
\end{tikzpicture}
\caption{Edge periodic cycle $\mathcal{G}$ (left) together with snapshots $\mathcal{G}(t)$ for $0 \leq t \leq 4$.}
\label{fig:example}
\end{figure}
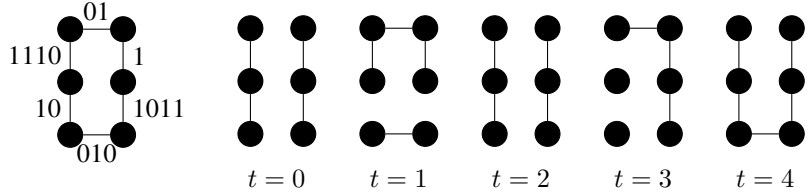
\section{It's hard to run around a table}
\label{sec:hardness}
In this section, we show that the \NP-hardness of \CNR already holds if the input graphs are very restricted. More precisely, we show that \CNR is \NP-hard and \W1-hard when parameterized by $|G|$, even for directed and undirected edge periodic cycles $\mathcal{G}$.
\begin{theorem}
	\label{thm:CycleNP-hard}
	\CNR on directed or undirected edge periodic cycles is~\NP-hard, and \W1-hard parameterized by the size of the underlying graph $G$. 
\end{theorem}
Both, the undirected and directed case of~\Cref{thm:CycleNP-hard} is shown by a reduction of the \PCA problem which was shown to be both \NP-hard and \W{1}-hard when parameterized by $|X|$ in~\cite{DBLP:conf/mfcs/MorawietzRW20}.
\prob{\PCA}{A finite set~$X\subseteq \{0,1\}^*$ of binary strings.}{Is there a position~$i$, such that~$\modid{x}{i} = 1$ for all~$x\in X$, \\where~$\modid{x}{i} := x[i \mod |x|]$?}
We begin with considering the case of undirected edge periodic cycles and then proceed by adapting the obtained construction for directed edge periodic cycles. 
\begin{lemma}
\label{thm:undirectedNP-hard}
\CNR on undirected edge periodic cycles is~\NP-hard and \W1-hard parameterized by the size of the underlying graph $G$. 
\end{lemma}
\begin{proof}
%
We first sketch the idea of the construction. It is helpful to consider Figure~\ref{fig:NP-hard} in the following. We represent each string in $X$ by an edge label. The constructed cycle will consist of two chains connected by two special edges. 
In the first chain, the elements in $X$ are increasingly listed in some fixed order as individual edge labels each. In the second chain the same edge labels are listed decreasingly in the same order. This will allow the cop and the robber to occupy antipolar vertices with the same edge labels on incident edges. Hence, while the cop is on one chain and the robber on the other chain, the robber can mimic the movements of the cop.
The two chains are connected by two special edges for which their edge labels are complementary in one position of the labels and identical in all other positions. This will allow the cop to switch between the chains in a certain time step while the robber is trapped on his chain. In this situation, the cop will be able to catch the robber if and only if there is a position~$i$, such that~$\modid{x}{i} = 1$ for all~$x\in X$, in which case all edges of the chains will be present for some period.

\begin{figure}
	\centering
	\begin{tikzpicture}[yscale=0.9]
		\tikzset{mystyle/.style={state, inner sep=2pt,minimum size=1pt, fill=black, node distance=30pt}}
		\node[mystyle] (0) {};
		\node[mystyle,right of=0] (1) {};
		\node[mystyle,right of=1] (2) {};
		\node[right of=2] (3) {$\dots$};
		\node[mystyle,right of=3] (4) {};
		\node[mystyle,right of=4] (5) {};
		\node[right of=5] (6) {$\dots$};
		\node[mystyle,right of=6] (7) {};
		\node[mystyle,right of=7] (8) {};
		
		\node[mystyle,below of=8,node distance=20pt] (a) {};
		\node[mystyle,left of=a] (b) {};
		\node[mystyle,left of=b] (c) {};
		\node[left of=c] (d) {$\dots$};
		\node[mystyle,left of=d] (e) {};
		\node[mystyle,left of=e] (f) {};
		\node[left of=f] (g) {$\dots$};
		\node[mystyle,left of=g] (h) {};
		\node[mystyle,left of=h] (i) {};
		
		\path
		(0) edge node[above] {$\xi(x_1)$} (1)
		(1) edge node[above] {$\xi(x_2)$} (2)
		(2) edge (3)
		(3) edge (4)				
		(4) edge node[above] {$\xi(x_i)$} (5)
		(5) edge (6)
		(6) edge (7)
		(7) edge node[above] {$\xi(x_m)$} (8)
		(8) edge node[right] {$0{0^m}1{0^m}$} (a)
		(a) edge node[below] {$\xi(x_1)$} (b)
		(b) edge node[below] {$\xi(x_2)$} (c)
		(c) edge (d)
		(d) edge (e)				
		(e) edge node[below] {$\xi(x_i)$} (f)
		(f) edge (g)
		(g) edge (h)
		(h) edge node[below] {$\xi(x_m)$} (i)
		(i) edge node[left] {$1{0^m}1{0^m}$} (0)
		;
	\end{tikzpicture}
	\caption{\CNR instance constructed from a \PCA instance with set of strings $X =\{x_1, \dots, x_m\}$ in the proof of~\Cref{thm:undirectedNP-hard}. For $x_j\in X$ the edge labels are defined as $\xi(x_j) := \xi(x_j[0]) \cdot \xi(x_j[1]) \cdot \ldots \cdot \xi(x_j[|x_j|-1])$, with $\xi(c) := 0{c^m}0{1^m}$ for $c\in\{0,1\}$. The upper chain corresponds to the vertices $r_j$ and the lower chain to the vertices $\ell_j$.}
	\label{fig:NP-hard}
\end{figure}
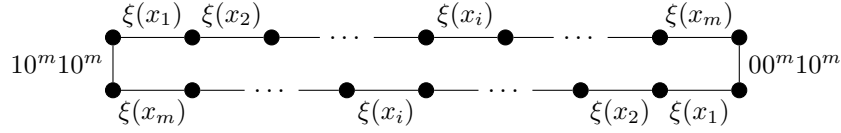

We now proceed with the formal proof.
Let~$X$ be an instance of~\PCA.
We describe how to construct in polynomial time an instance~$\mathcal{G}=(V,E, \tau)$ of~\CNR, where~$\mathcal{G}$ is an undirected edge periodic cycle such that~$X$ is a yes-instance of~\PCA if and only if~$\mathcal{G}$ is a yes-instance of~\CNR.

Let~$|X| = m$ and $\{x_1, \dots, x_{m}\}$ be the elements of~$X$.
We set~$V := \{\ell_j, r_j \mid 0 \leq j \leq m\}$ and~$E:=\{\{\ell_{j-1}, \ell_j\}, \{r_{j-1}, r_j\} \mid 1  \leq j \leq m\} \cup \{\{\ell_0, r_{m}\}, \{\ell_{m}, r_0\}\}$.
Next, we set~$\tau(\{\ell_0, r_{m}\}) := 10^{m}10^{m}$ and~$\tau(\{\ell_{m}, r_0\}) := 00^{m}10^{m}$.
Let~$\xi(c) := 0c^{m}01^{m}$ for all~$c \in\{0,1\}$.
Finally, we set~$\tau(\{\ell_{j-1}, \ell_j\}) := \tau(\{r_{j-1}, r_j\}) := \xi(x_j[0]) \cdot \xi(x_j[1]) \cdot \ldots \cdot \xi(x_j[|x_j|-1])$ for each~$x_j\in X$.
Note that the length of each edge label is divisible by~$q := 2m+2$.
For~$i \geq 0$, let~$T_i := [q \cdot i, q\cdot (i+1)-1]$ denote the~\emph{$i$-th time block}, that is, the~$q$ consecutive time steps starting from~$q\cdot i$.
Note that the~$j$-th edge label limited to the~$i$-th time block~$\modid{\tau(\{\ell_{j-1}, \ell_j\})}{T_i} = \modid{\tau(\{r_{j-1}, r_j\})}{T_i}$ is exactly~$\xi(\modid{x_j}{i})$.

Next, we show that~$X$ is a yes-instance of~\PCA if and only if~$\mathcal{G}$ is a yes-instance of~\CNR.

$(\Rightarrow)$
Let~$i$ be a position such that~$\modid{x}{i} = 1$ for all~$x\in X$.
We describe the winning strategy for the cop.
She should choose the vertex~$\ell_0$ as her start vertex and should never move until the beginning~$t$ of the~$i$-th time block~$T_i$.
Since~$\modid{x}{i} = 1$ for all~$x\in X$,~$\modid{\tau(\{\ell_{j-1}, \ell_j\})}{T_i} = \modid{\tau(\{r_{j-1}, r_j\})}{T_i} = \xi(1) = 01^m01^m$.
Consequently, in time step~$t$ only the edge~$\{\ell_0,r_m\}$ exists and in the following~$m$ time steps, all edges except~$\{\ell_0,r_m\}$ and~$\{\ell_m,r_0\}$ exist.

If the robber is currently on some vertex~$r_j$, then the cop should move to~$r_{m}$ in time step~$t$.
Otherwise, the cop should stay on~$\ell_0$ in this time step.
By the fact that the edge~$\{\ell_{m}, r_0\}$ does not exist in~time step~$t$, we obtain that, at the beginning of time step~$t + 1$, both players are either on vertices labeled with~$r$ or labeled with~$\ell$.
Since all edges of the two paths~$(\ell_0, \dots, \ell_{m})$ and~$(r_0, \dots, r_{m})$ exist in the time steps~$[t +1, t+m]$, the cop can catch the robber in at most~$m$ time steps, since neither~$\{\ell_0, r_{m}\}$ nor~$\{\ell_{m}, r_0\}$ exists in any of the time steps~$[t +1, t+m]$.
Consequently,~$\mathcal{G}$ is a yes-instance of~\CNR.

$(\Leftarrow)$
Suppose that~$X$ is a no-instance of~\PCA.
We describe a winning strategy for the robber.
In the following, we say that the vertex~$\ell_j$ is the~\emph{mirror vertex} of~$r_j$ and vice versa.
Moreover, we say that the robber \emph{mirrors the move} of the cop at some time step~$t$, if the cop is on the mirror vertex of the robber at the beginning of time step~$t$ and the robber moves to the mirror vertex of the vertex the cop ends on in time step~$t$. 

The start vertex of the robber should be the mirror vertex of the start vertex of the cop.
If it is possible, then the robber should always mirror the moves of the cop.

Note that the only move the robber \emph{cannot} mirror is if the cop traverses the edge~$\{\ell_{m}, r_0\}$ at the beginning of some~$i$-th time block.

We show that the robber has a strategy to end on the mirror vertex during the~$i$-th time block and evade the cop until then.

Assume w.l.o.g., that the cop moves from~$\ell_{m}$ to~$r_0$ and, thus, the robber is currently on~$r_{m}$.
Since~$X$ is a no-instance of~\PCA, there is at least one~$x_j\in X$ with~$\modid{x_j}{i} = 0$.
Hence,~$\modid{\tau(\{\ell_{j-1}, \ell_j\})}{T_i} = \modid{\tau(\{r_{j-1}, r_j\})}{T_i} = \xi(0) = 00^m01^m$.
Consequently, the cop cannot catch the robber in the first~$m+1$ time steps of the~$i$-th time block.
Hence, the robber should stay on this vertex until the beginning of time step~$q \cdot i + m + 1$.

If the cop moves from~$r_0$ to~$\ell_{m}$ in time step~$q \cdot i + m + 1$, the robber is again on the mirror vertex of the cop and is able to mirror all of the cops moves in the remaining time steps of this time block.
Otherwise, the cop stays on some vertex~$r_p$.
In this case, the robber should move to~$\ell_0$.
Since the edges~$\{\ell_0, r_m\}$ and~$\{\ell_m, r_0\}$ do not exist in the remaining time steps of this time block, the cop cannot catch the robber in this time block.
Moreover, since all edges of the path~$(\ell_0, \dots, \ell_{m})$ exist in the last~$m$ time steps of the~$i$-th time block, the robber can move along the path~$(\ell_0, \dots, \ell_{m})$ and reach the mirror vertex of the cop in at most~$m$ time steps.
Consequently, we can show via induction, that the robber has an infinite evasive strategy and, thus,~$\mathcal{G}$ is a no-instance of~\CNR.
Since \PCA is~\W1-hard when parameterized by~$|X|$ and~$|V| = |E| = 2\cdot |X| + 2$, we obtain that~\CNR is \W1-hard when parameterized by the size of the underlying graph of $\mathcal{G}$ even on undirected edge periodic cycles.\lncsqed
\end{proof}

Next, we adapt the previous construction for \emph{directed} edge periodic cycles. It is helpful to consider Figure~\ref{fig:NP-hard dir} in the following.
In the adaption, we only have one chain listing the elements of $X$. The end vertex of this chain is connected to a new vertex $s$ which is again connected to the start vertex of the chain. The edges incident with $s$ will act as the two edges connecting the two chains in the previous construction by delaying the robber, such that the cop can catch him if all edges corresponding to $X$ are present in some time period.
\begin{lemma}\label{thm:dir-NP}
	\CNR on \emph{directed} edge periodic cycles is~\NP-hard, and \W1-hard parameterized by the size of the underlying graph.
\end{lemma}
\begin{proof}

Again, we reduce from \PCA.
Let~$X$ be an instance of~\PCA.
We describe how to construct an instance~$\mathcal{G}=(V,E, \tau)$ of~\CNR, where~$\mathcal{G}$ is a directed edge periodic cycle.
%
Let~$|X|=m$ and $\{x_1, \dots, x_{m}\}$ be the elements of~$X$.
We set~$V := \{v_j \mid 0 \leq j \leq m\} \cup \{s\}$ and~$E:=\{(v_{j-1}, v_j) \mid 1  \leq j \leq m\} \cup \{(v_{m}, s), (s,v_{0})\}$.
Next, we set~$\tau((v_{m}, s)) := 0^{m}10^{m}0$ and~$\tau((s, v_{0})) := 0^{m}00^m1$.
Let~$\xi(c) := c^{m}01^{m}1$ for all~$c \in\{0,1\}$.
Finally, we set~$\tau((v_{j-1}, v_j)) := \xi(x_j[0]) \cdot \xi(x_j[1]) \cdot \ldots \cdot \xi(x_j[|x_j|-1])$ for each~$x_j\in X$.

Note that the length of each edge label is divisible by~$q := 2m+2$.
For~$t \geq 0$, let~$T_t := [q \cdot t, q\cdot (t+1) -1]$ denote the~\emph{$t$-th time block}, that is, the~$q$ consecutive time steps starting from~$q\cdot t$.
Note that the~$j$-th edge label limited to the~$t$-th time block~$\modid{\tau((v_{j-1},v_j\})}{T_t}$ is exactly~$\xi(\modid{x_j}{t})$.
Next, we show that~$X$ is a yes-instance of~\PCA if and only if~$\mathcal{G}$ is a yes-instance of~\CNR.

$(\Rightarrow)$ 
Let~$t$ be a position such that~$\modid{x}{t} = 1$ for all~$x\in X$.
We describe the winning strategy for the cop.
The cop should choose the vertex~$v_0$ as her start vertex and should never move until the beginning~$t^* := q\cdot t$ of the~$t$-th time block.
By construction and the fact that~$\modid{x_i}{t}=1$ for each~$x_i\in X$,~$\modid{\tau((v_{i-1}, v_{i}))}{T_t} = \xi(1) = 1^m01^m1$.
Hence, the cop can move from vertex~$v_i$ to vertex~$v_{i+1}$ in time step~$t^* + i$ for each~$i\in[0, m-1]$ and, thus, reach the vertex~$v_{m}$ in time step~$t^* + m-1$.
Moreover, the cop can then move to the vertex~$s$ in time step~$t^* + m$.
By construction,~$\modid{\tau((s, v_0))}{t^* +j} = 0$ for each~$j\in[0, m]$.
Hence, the cop has a winning strategy since she started at vertex~$v_0$ and moved over every vertex of~$V$ while the robber was not able to traverse the edge~$(s, v_0)$.

$(\Leftarrow)$
Suppose that for every position~$t$, there is some~$x_j\in X$ with~$\modid{x_j}{t} = 0$.
We show that the robber has a winning strategy.
%
For some time step, let $w_C$ and $w_R$ denote the vertex of the cop, respectively robber in this time step.
We call the vertex~$v_0$ \emph{safe} for all vertices of~$V\setminus \{v_0, s\}$, we call~$v_{m}$~\emph{safe} for~$v_0$ and~$s$, and we call~$s$~\emph{safe} for~$v_0$.
Let~$u_C$ be the start vertex of the cop, then the robber should choose a vertex which is safe for~$u_C$ as his start vertex.
\begin{figure}
	\centering
	\begin{tikzpicture}[yscale=0.9]
		\tikzset{mystyle/.style={state, inner sep=2pt,minimum size=1pt, fill=black, node distance=30pt}}
		\node[mystyle] (0) {};
		\node[mystyle,right of=0] (1) {};
		\node[mystyle,right of=1] (2) {};
		\node[right of=2] (3) {$\dots$};
		\node[mystyle,right of=3] (4) {};
		\node[mystyle,right of=4] (5) {};
		\node[right of=5] (6) {$\dots$};
		\node[mystyle,right of=6] (7) {};
		\node[mystyle,right of=7] (8) {};
		
		\node[mystyle,above of=4,node distance=25pt] (s) {};
		\node[above of=s,node distance=10pt] (s2) {$s$};
		
		\path
		(0) edge[->, >=latex] node[below] {$\xi(x_1)$} (1)
		(1) edge[->, >=latex] node[below] {$\xi(x_2)$} (2)
		(2) edge (3)
		(3) edge[->, >=latex] (4)				
		(4) edge[->, >=latex] node[below] {$\xi(x_i)$} (5)
		(5) edge (6)
		(6) edge[->, >=latex] (7)
		(7) edge[->, >=latex] node[below] {$\xi(x_m)$} (8)
		(8) edge[->, >=latex,bend angle=15, bend right] node[above] {${0^m}1{0^m}0$} (s)
		(s) edge[->, >=latex,bend angle=15, bend right] node[above] {${0^m}0{0^m}1$} (0)
		;
	\end{tikzpicture}
	\caption{\CNR instance constructed from a \PCA instance with set of strings $X =\{x_1, \dots, x_m\}$ in the proof of~\Cref{thm:dir-NP}. For $x_j\in X$ the edge labels are defined by the homomorphism $\xi(x_j) := \xi(x_j[0]) \cdot \xi(x_j[1]) \cdot \ldots \cdot \xi(x_j[|x_j|-1])$, with $\xi(c) := {c^m}0{1^m}1$ for $c\in\{0,1\}$.}
	\label{fig:NP-hard dir}
\end{figure}
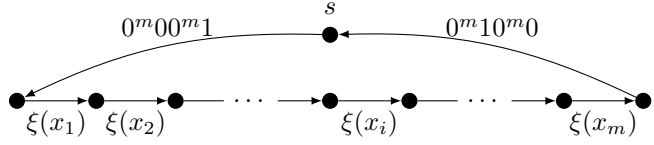
\begin{myclaim}\label{claim:directed robber strategy}
Let~$t^* = t \cdot q$ be the beginning of the~$t$-th time block for some $t\geq 0$, let~$u_C$ be the vertex of the cop at time step~$t^*$ and~$u_R$ be the vertex of the robber at time step~$t^*$.
If~$u_R$ is safe for~$u_C$, then the robber has a strategy such that the cop cannot catch him in the~$t$-th time block and the robber ends on a vertex that is safe for the vertex of the cop at the end of the~$t$-th time block.
\end{myclaim}
\begin{claimproof}
\textbf{Case 1:}~$u_C \in V \setminus \{s, v_0\}$ and~$u_R = v_0$\textbf{.}
The robber should wait on vertex~$v_0$ until the beginning of time step~$t^* + m$.
Since the edge~$(s, v_0)$ only exists in the last time step of the~$t$-th time block, the cop cannot catch the robber in any of these time steps.
If the cop does not traverse the edge~$(v_{m}, s)$ in time step~$t^* + m$, then the robber should stay on vertex~$v_0$ until the beginning of the next time block.
Since the edge~$(v_{m}, s)$ only exists in time steps~$t'$ with~$t' \mod q = m$, it follows that the cop ends on some vertex of~$V \setminus \{s, v_0\}$ at the end of the~$t$-th time block.
Thus, at the beginning of the~$(t+1)$-th time block, the vertex of the robber is safe for the vertex of the cop.

Otherwise, the cop traverses the edge~$(v_{m}, s)$ in time step~$t^* + m$. 
Then, the robber should traverse the edge~$(v_{i-1},v_{i})$ in time step~$t^* + m + i$ for each~$i\in [1, m]$, while the cop has to wait on $s$.
Hence, the robber reaches~$v_{m}$ in time step~$t^* + q - 2$.
In time step~$t^* + q - 1$, the cop can either stay on~$s$ or move to~$v_0$.
In both cases the robber should stay on~$v_{m}$ which is safe for both~$s$ and~$v_0$.

\textbf{Case 2:}~$u_C = s$ and~$u_R = v_{m}$\textbf{.}
Since the edge~$(s, v_0)$ only exists in the last time step of the~$t$-th time block, the cop has to stay on~$s$ until the beginning of time step~$t^* + q - 1$.
In time step~$t^* + q - 1$, the cop can either stay on~$s$ or move to~$v_0$.
In both cases the robber stays on~$v_{m}$ which is safe for both~$s$ and~$v_0$.

\textbf{Case 3:}~$u_C = v_0$ and~$u_R \in \{v_m, s\}$\textbf{.}
Let~$j\in [1, m]$ such that~$\modid{x_j}{t} = 0$, recall that by definition of~$\tau$ it follows that~$\modid{\tau((v_{j-1}, v_j))}{T_t} = 0^m01^m1$. 
Thus, the cop cannot reach the vertex~$v_{m}$ in the first~$m+1$ time steps of the~$t$-th time block.
In time step~$t^*+m$, the robber should stay on~$s$ if~$s$ is his current vertex or traverse the edge~$(v_{m}, s)$, otherwise.
Since this is the only time step in which this edge exists in the~$t$-th time block, the cop cannot catch the robber in this time block.
Until the beginning of time step~$t^* + q -1$, the robber should stay on~$s$.
If the cop ends her turn on vertex~$v_0$, then the robber should stay on~$s$.
Otherwise, the robber should traverse the edge~$(s, v_0)$ in time step~$t^* + q -1$.
In both cases, the vertex of the robber is safe for the vertex of the cop at the beginning of the~$(t+1)$-th time block.
\end{claimproof}

By using Claim~\ref{claim:directed robber strategy}, one can show via induction, that the robber has an infinite evasive strategy and, thus,~$\mathcal{G}$ is a no-instance of~\CNR.
%
\end{proof}

For the next section, we will stick to edge periodic cycles and consider families of cop-winning \emph{undirected} edge periodic cycles.


\section{Sharp bounds on the length required to ensure robber-winning edge periodic cycles}
\label{sec:lowerBounds}
In~\cite{erlebach2020game}, an upper bound on the cycle length of an edge periodic cycle in dependence of $\LCM(L_{\mathcal{G}})$, required to ensure an robber winning strategy, was given.
Namely, for $|V|=n$, the graph $\mathcal{G}$ is robber winning if $n \geq 2 \cdot \ell \cdot \LCM(L_{\mathcal{G}})$, where $\ell = 1$ if $\LCM(L_{\mathcal{G}}) \geq 2\cdot \max(L_{\mathcal{G}})$, and $\ell = 2$, otherwise (\cite[Theorem 3]{erlebach2020game}).
So far, these bounds where not sharp, as for instance, in~\cite{erlebach2020game}, the only lower bounds are given by cop winning strategies for families of edge periodic cycles with $n=1.5\cdot\LCM(L_{\mathcal{G}})$ for $\ell =1$ (\cite[Theorem 5]{erlebach2020game}), and $n=3\cdot \LCM(L_{\mathcal{G}})$ for $\ell=2$ and $\max(L_{\mathcal{G}})=\LCM(L_{\mathcal{G}})$ (\cite[Theorem 4]{erlebach2020game}).
We show that both upper bounds ensuring a robber winning strategy are sharp by presenting infinite families of cop-winning edge periodic cycles with $n=2\cdot \ell\cdot\LCM(L_{\mathcal{G}})-1$ vertices.
\begin{theorem}
	\label{thm:shar-bounds}
	For every $k\geq 3$ and~$\ell\in\{1,2\}$,
	 there exists a cop-winning edge periodic cycle $\mathcal{G}=(V, E, \tau)$ with 
	$\max(L_{\mathcal{G}}) = k$ and 
	$n =  2\cdot \ell \cdot \LCM(L_{\mathcal{G}})-1$ vertices, where~$\LCM(L_{\mathcal{G}}) \geq 2k$ if~$\ell = 1$, and~$\LCM(L_{\mathcal{G}}) = k$, otherwise.
\end{theorem}
In order to prove~\Cref{thm:shar-bounds}, we give families of edge periodic cycles for $\ell = 1$ and $\ell = 2$, each, beginning with $\ell = 2$, i.e., the case that $\LCM(L_{\mathcal{G}}) < 2\cdot \max(L_{\mathcal{G}})$.
\begin{lemma}
	\label{lem:famCopWin}
	For every $k\geq 2$ there exists an edge periodic cycle $\mathcal{G}=(V, E, \tau)$ with $\LCM(L_{\mathcal{G}})=k = \max(L_{\mathcal{G}})$, and $n =  4k-1$ vertices that is cop-winning.
\end{lemma}
\begin{proof}
	Consider the edge periodic cycle $\mathcal{G}_k=(V, E, \tau)$ depicted in Figure~\ref{fig:Cop-Win-4k-1} with $|V| = 4k-1$. 
	This graph admits a cop winning strategy if the cop picks the highlighted vertex with index 0 as her start vertex.
	The vertices are indexed by positive numbers indicating their clockwise distance to the start vertex of the cop, and with negative numbers indicating their counterclockwise distance.
	Let the cop pick vertex $0$. We consider the antipolar vertices $+(2k-1)$ and $-(2k-1)$ as potential start vertices of the robber. 
	We show that if the robber picks vertex $+(2k-1)$, then the cop has a winning strategy by continuously running clockwise, starting in time step zero, and if the robber picks vertex $-(2k-1)$, the same applies running counterclockwise. Note that these two positions represent extrema and being able to catch the robber at these vertices implies being able to catch him at all vertices in the graph. Table~\ref{tab:4k-1} shows the positions of the cop and robber for these strategies for $k\geq 4$. For each time step, the position after both players moved are depicted; $s$ is the start configuration. We abbreviate consecutive 1-edges and only depict the time steps and positions when one of the players reaches a non-trivial edge.
	For the cases of $k=2$ and $k=3$ the cop catches the robber earlier than depicted in Table~\ref{tab:4k-1}, namely in step $t=6$ clockwise and $t=8$ counterclockwise for~$k=2$ and in step $t=6$ clockwise and $t=9$ counterclockwise for~$k=3$ if the robber chooses the corresponding antipolar start vertices.
	Details on case $k=2$ and $k=3$, and a concrete example for $k=4$, can be found in the appendix.
	\lncsqed
\end{proof}
\begin{figure}[t]
	\centering
	\begin{tikzpicture}[node distance=15pt,yscale=0.9]
		\node[state,fill=red,minimum size=3pt] (0) at (5, 2) {C};
		\node[below of=0] {0};
		\node[state,fill=black,minimum size=3pt] (1) at (7, 2) {};
		\node[below of=1] (1'){+1};
		\node[below of=1'] {-(4$k$-2)};
		\node[state,fill=black,minimum size=3pt] (3) at (4, 0) {};
		\node[below of=3] (3') {+(2$k$)};
		\node[below of=3'] {-(2$k$-1)};
		\node[state,fill=black,minimum size=3pt] (2) at (6, 0) {};
		\node[below of=2] (2') {+(2$k$-1)};
		\node[below of=2'] {-(2$k$)};
		\node[state,fill=black,minimum size=3pt] (5) at (1, 1.5) {};
		\node[left of=5,node distance=2.5em] (5') {+($3k$-1)};
		\node[below of=5'] {-$k$};
		\node[state,fill=black,minimum size=3pt] (4) at (1, 0.5) {};
		\node[below of=4] (4') {+($3k$-2)};	
		\node[below of=4'] {-($k$+1)};	
		\path (0) edge node[above,fill=gray!20] {$10^{k-1}$} (1)
		(1) edge[dashed,out=0,in=90] (8,1) 
		(8,1) edge[dashed,out=-90,in=0](2)
		(2) edge node[above,fill=gray!20] {$0^{k-1}1$} (3)
		(3) edge[dashed,in=-30,out=180] (4)
		(4) edge node[right,fill=gray!20] {$10^{k-1}$} (5)
		(5) edge[dashed,out=30,in=180] (0)
		(0) edge[thick, right,fill=gray!20] (1)
		(2) edge[thick, right,fill=gray!20] (3)
		(4) edge[thick, right,fill=gray!20] (5)
		; 
	\end{tikzpicture}
\caption{Cycle with $4\cdot k-1$ vertices and $\LCM(L_{\mathcal{G}})=k$ with a cop winning strategy from the start vertex marked in red. Edges not drawn (depicted by dots) are 1-edges; for all other edges, $\tau(e)$ is explicitly noted (with gray background). The clockwise [counterclockwise] distance of each vertex to the start vertex of the cop is given as a positive [negative] number.}
\label{fig:Cop-Win-4k-1}
\end{figure}
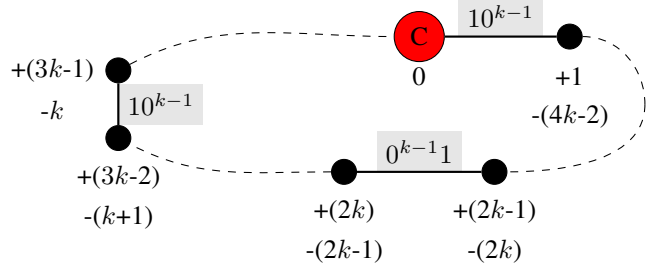
\begin{table}[t!]
	\centering
	\begin{tabular}{r || r | r}
		time step & pos.~cop & pos.~robber\\
		\hline
		$s$ &  $0$ & $2k-1$ \\
		$0$ & $1$ & $2k-1$ \\ 
		$k-1$ & $k$ & $2k$ \\
		$2k-3$ & $2k-2$ & $3k-2$ \\
		$2k-2$ & $2k-1$ & $3k-2$ \\
		$2k-1$ & $2k$ & $3k-2$ \\
		$2k$ & $2k+1$ & $3k-1$ \\
		$3k-3$ & $3k-2$ & $4k-4$ \\
		$3k$ & $3k-1$ & $0$ \\
		$3k+1$ & $3k$ & $0$ \\
		$4k-1$ & $4k-2$ & $0$ \\
		$4k$ & $0$ & $\skull$
	\end{tabular}
	\hfil
	\begin{tabular}{r || r | r}
		time step & pos.~cop & pos.~robber\\
		\hline
		$s$ & $0$ & $-(2k-1)$ \\
		$k-1$ & $-(k)$ & $-(2k)$ \\
		$k$ &  $-(k+1)$ & $-(2k+1)$ \\
		$2k-2$ & $-(2k-1)$ & $-(3k-1)$ \\
		$2k-1$ & $-(2k)$ & $-(3k)$ \\
		$3k-3$ & $-(3k-2)$ & $-(4k-2)$ \\
		$3k$ & $-(3k+1)$ & $0$ \\
		$4k-3$ & $-(4k-2)$ & $-(k-3)$ \\
		$4k$ & $0$ & $-(k)$ \\
		$5k-1$ & $-(k-1)$ & $-(k)$ \\
		$5k$ & $-(k)$ & $\skull$\\
		\multicolumn{3}{l}{\ } 
	\end{tabular}
	\caption{Time steps with corresponding positions of cop and robber in the edge periodic cycle depicted in Figure~\ref{fig:Cop-Win-4k-1}. All positions are \emph{after} moving in this time step. The time step $s$ denotes the start configuration. Recall that the cop moves first.}
	\label{tab:4k-1}
\end{table}
For the case that $\ell = 1$, i.e., when $\LCM(L_{\mathcal{G}}) \geq 2 \cdot \max(L_\mathcal{G})$, we slightly adapt the family of graphs depicted in Figure~\ref{fig:Cop-Win-4k-1}.
Note that for $\max(L_{\mathcal{G}})=2$ there is no edge periodic cycle $\mathcal{G}=(V, E, \tau)$ with $\LCM(L_{\mathcal{G}}) > \max(L_{\mathcal{G}}) = 2$.
\begin{lemma}
	For every $k\geq 3$ with~$k\neq 2^m$ for all $m \in \mathbb{N}$, there exists an edge periodic cycle $\mathcal{G}=(V, E, \tau)$ with $\LCM(L_{\mathcal{G}}) = 2\cdot \max(L_{\mathcal{G}}) = 2\cdot k$, and $n = 2\cdot 2k-1$ vertices that is cop-winning.
\end{lemma}
\begin{proof}
	For the case $\ell =1$ we introduce an artificial edge label in the edge periodic cycle in Figure~\ref{fig:Cop-Win-4k-1}, such that the $\LCM(L_{\mathcal{G}})$ is exactly $2k$. This edge will not affect the run of the cop. 
	Its purpose is to introduce a factor $2$ in the number of vertices compensating the missing factor $2$ from the variable $\ell$. 
	Therefore, note that the edge $e_{1,2}$ connecting vertex +$1$ and +$2$ is taken by the cop only once, in the clockwise run in time step $1$ and in the counterclockwise run in time step $4k-3$. 
	Hence, the cop only crosses the edge in an odd time step. 
	We can write $k$ as $k = 2^i \cdot j$ where $j$ is an odd number with~$j > 1$ since~$k \neq 2^m$. 
	Then, introducing a string $\tau(e_{1,2}) = 01^{2^{i+1}-1}$ of length $2^{i+1}$ yields a least common multiple of $\LCM(L_{\mathcal{G}}) = 2^{i+1} \cdot j = 2 \cdot k$. \lncsqed
\end{proof}

In the case of $\max(L_{\mathcal{G}}) = k=2^m$ for some $m \in \mathbb{N}$, it holds that for the smallest possible value of $\LCM(L_{\mathcal{G}})$ with $\LCM(L_{\mathcal{G}})>\max(L_{\mathcal{G}})$, we have $\LCM(L_{\mathcal{G}}) \geq 3 \cdot \max(L_{\mathcal{G}})$. Hence, in these cases we need a separate family of graphs.

 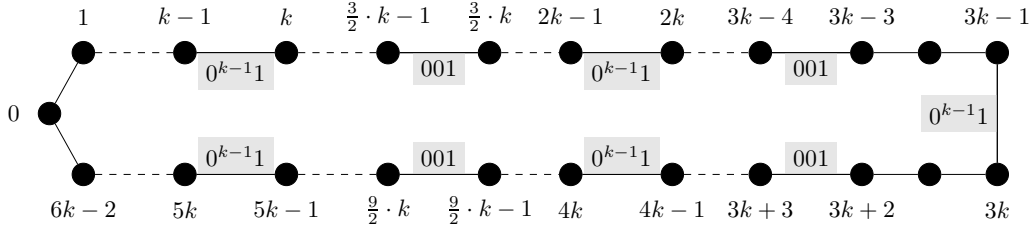
\begin{figure}
	\centering
		\scalebox{0.9}{
	\begin{tikzpicture}[node distance=15pt,,yscale=0.9, state/.style={circle, draw,minimum size=20pt, initial text = }]

		\node[state,fill=black,minimum size=3pt] (0) at (0, 1) {};
		
		\node[state,fill=black,minimum size=3pt] (o1) at (0.5, 2) {};
		\node[state,fill=black,minimum size=3pt] (u1) at (0.5, 0) {};	\draw (o1) edge (0);
\draw (u1) edge (0);

		\node[state,fill=black,minimum size=3pt] (o2) at (2, 2) {};	
		\node[state,fill=black,minimum size=3pt] (u2) at (2, 0) {};
		
		\node[state,fill=black,minimum size=3pt] (o3) at (3.5, 2) {};
		\node[state,fill=black,minimum size=3pt] (u3) at (3.5, 0) {};
\draw (u2) edge node[above,fill=gray!20] {$0^{k-1}1$} (u3);	
\draw (o2) edge node[below,fill=gray!20] {$0^{k-1}1$} (o3);	
\draw (u2) edge (u3);	
\draw (o2) edge (o3);

		\node[state,fill=black,minimum size=3pt] (o6) at (7.7, 2) {};
		\node[state,fill=black,minimum size=3pt] (u6) at (7.7, 0) {};	
		
		\node[state,fill=black,minimum size=3pt] (o7) at (9.2, 2) {};
		\node[state,fill=black,minimum size=3pt] (u7) at (9.2, 0) {};	
		
\draw (u6) edge node[above,fill=gray!20] {$0^{k-1}1$} (u7);	
\draw (o6) edge node[below,fill=gray!20] {$0^{k-1}1$} (o7);	
\draw (u6) edge (u7);	
\draw (o6) edge (o7);

		\node[state,fill=black,minimum size=3pt] (o4) at (5, 2) {};
		\node[state,fill=black,minimum size=3pt] (u4) at (5, 0) {};
		
		\node[state,fill=black,minimum size=3pt] (o5) at (6.5, 2) {};
		\node[state,fill=black,minimum size=3pt] (u5) at (6.5, 0) {};
\draw (u4) edge node[above,fill=gray!20] {$001$} (u5);	
\draw (o4) edge node[below,fill=gray!20] {$001$} (o5);	
\draw (u4) edge (u5);	
\draw (o4) edge (o5);

		\node[state,fill=black,minimum size=3pt] (o8) at (10.5, 2) {};
		\node[state,fill=black,minimum size=3pt] (u8) at (10.5, 0) {};
		
		\node[state,fill=black,minimum size=3pt] (o9) at (12, 2) {};
		\node[state,fill=black,minimum size=3pt] (u9) at (12, 0) {};
		\node[state,fill=black,minimum size=3pt] (o10) at (13, 2) {};
		\node[state,fill=black,minimum size=3pt] (o11) at (14, 2) {};
		\node[state,fill=black,minimum size=3pt] (u10) at (13, 0) {};
		\node[state,fill=black,minimum size=3pt] (u11) at (14, 0) {};
\draw (u8) edge node[above,fill=gray!20] {$001$} (u9);	
\draw (o8) edge node[below,fill=gray!20] {$001$} (o9);	
\draw (o11) edge node[left,fill=gray!20] {$0^{k-1}1$} (u11);	
\draw (u8) edge (u9);
\draw (u10) edge (u9);
\draw (u10) edge (u11);	
\draw (o8) edge (o9);	
\draw (o9) edge (o10);	
\draw (o10) edge (o11);	
\draw (u11) edge (o11);

		\draw (o1) edge[dashed] (o2);
		\draw (o3) edge[dashed] (o4);
		\draw (o5) edge[dashed] (o6);
		\draw (o7) edge[dashed] (o8);
		
		\draw (u1) edge[dashed] (u2);
		\draw (u3) edge[dashed] (u4);
		\draw (u5) edge[dashed] (u6);
		\draw (u7) edge[dashed] (u8);

		\node[left of=0] {$0$};	
		\node[above of=o1] {$1$};	
		\node[below of=u1] {$6 k - 2$};
		\node[above of=o2] {$k - 1$};
		\node[below of=u2] {$5k$};
		\node[above of=o3] {$k$};	
		\node[below of=u3] {$5k-1$};
		\node[above of=o6] {$2k - 1$};	
		\node[below of=u6] {$4k$};
		\node[above of=o7] {$2k$};	
		\node[below of=u7] {$4k-1$};
		\node[above of=o4] {$\frac{3}{2}\cdot k - 1$};	
		\node[below of=u4] {$\frac{9}{2}\cdot k$};	
		\node[above of=o5] {$\frac{3}{2}\cdot k$};	
		\node[below of=u5] {$\frac{9}{2}\cdot k - 1$};
		\node[above of=o8] {$3k - 4$};
		\node[above of=o9] {$3k - 3$};
		\node[below of=u8] {$3k + 3$};
		\node[below of=u9] {$3k + 2$};	
		\node[above of=o11] {$3k-1$};
		\node[below of=u11] {$3k$};

	\end{tikzpicture}
}
		
	\caption{Cycle with $6\cdot k-1$ vertices and $\LCM(L_{\mathcal{G}})=3k$ with a cop winning strategy from the start vertex~$0$ where~$k=2^m$ and~$m\geq 2$. 
	Edges not drawn (depicted by dots) or edges without an explicit label are 1-edges; for all other edges, $\tau(e)$ is explicitly noted (with gray background).}
	\label{fig:k2igeneral}
\end{figure}
\begin{lemma}\label{lem:famCopWin2i}
	For every $k = 2^m$ with~$m\geq 2$, there exists an edge periodic cycle $\mathcal{G}=(V, E, \tau)$ with $\LCM(L_{\mathcal{G}}) = 3\cdot \max(L_{\mathcal{G}}) = 3\cdot k$, and $n = 6\cdot k-1$ vertices that is cop-winning.
\end{lemma}
\begin{proof}
Consider the edge periodic cycle $\mathcal{G}_k=(V, E, \tau)$ depicted in Figure~\ref{fig:k2igeneral} with $|V| = 6k-1$. 
	This graph admits a cop winning strategy if the cop picks the highlighted vertex with index~0 as her start vertex.
	The vertices are indexed by positive numbers indicating their clockwise distance to the start vertex of the cop.
	Let the cop pick vertex $0$. 
	We show that if the robber picks vertex $3k-1$, then the cop has a winning strategy by continuously running clockwise, starting in time step zero.
	Since for each~$j$, starting from vertex~$0$, the label of the~$j$-th edge clockwise is equal to the label of the~$j$-th edge counterclockwise, the same applies running counterclockwise if the robber picks vertex $3k$. 
	Note that these two positions represent extrema and being able to catch the robber at these vertices implies being able to catch him at all vertices in the graph.
	Suppose that the robber picks vertex $3k-1$.
	Since~$k=2^m$ for some~$m\geq 2$,~$\frac{3}{2}\cdot k$ and~$\frac{9}{2}\cdot k$ are divisible by~$3$.
	Hence for each~$j\in[1,6k-3]$, the cop can traverse the edge~$\{j, j+1\}$ in time step~$j$ and, thus, reach the vertex~$5k-1$ in time step~$5k-2$. 
	We show that, starting from vertex~$3k-1$ and running clockwise, the robber cannot reach vertex~$5k$ prior than time step~$5k-1$.
	This then implies, that the cop catches the robber after at most~$5k-2$ time steps.
	Note that the first time the robber can traverse the edge~$\{3k-1, 3k\}$ is at time step~$k-1$.
	Hence, the robber reaches the vertex~$3k+2$ not prior than time step~$k+1$.
	Since~$k$ is not divisible by~$3$, the robber cannot traverse the edge~$\{3k+2, 3k+3\}$ in time step~$k+2$.
	Thus, the robber cannot reach the vertex~$4k-1$ prior than time step~$2k$ and consequently, he cannot traverse the edge~$\{4k-1, 4k\}$ prior than time step~$3k-1$.
	Hence, the robber reaches the vertex~$\frac{9}{2}k-1$ not prior than time step~$\frac{7}{2}k-2$.
	Since~$k$ is not divisible by~$3$, the robber cannot traverse the edge~$\{\frac{9}{2}k-1, \frac{9}{2}k\}$ in time step~$\frac{7}{2}k-1$.
	Thus, the robber cannot reach the vertex~$5k-1$ prior than time step~$4k$ and consequently, he cannot traverse the edge~$\{5k-1, 5k\}$ prior than time step~$5k-1$.
	Hence, the statement holds.	
	A concrete example for $k=4$ can be found in the appendix.
\end{proof}

\section{Complexity upper bounds}
\label{sec:upperbounds}
The main result of this section is that the \CNR problem for \emph{general} edge periodic graphs can be solved in polynomial space. Note that two-player games may take exponentially many turns and hence containment in \PSPACE is not obvious. In our case, already the period of graphs on which the game is played is exponential in general. 
This prohibits a standard incremental \PSPACE algorithm approach. 
We show that despite the potentially exponential period of the sequence of graphs $\mathcal{G}(t)$ we show that we can determine whether the cop has a winning strategy by sweeping through the configuration space in a way that we only consider polynomially many vertices in each step. The fact that we only consider one cop and one robber is here crucial for the polynomial bound.
\begin{theorem}
	\label{thm:PSPACE-Membership}
 \CNR for edge periodic graphs is contained in~\PSPACE.
\end{theorem}
For general edge periodic graphs, the \CNR problem was reduced in~\cite{erlebach2020game} to a variant of the \textsc{And-Or Graph Reachability} problem~\cite{DBLP:journals/tcs/AlfaroHK07} via an exponential time reduction. The \textsc{And-Or Graph Reachability} problem is a two player game where players move a token in an \textsc{And-Or} graph from a source to a target. An \textsc{And-Or} graph is a graph $G=(V_\wedge \cup V_\vee,E)$ where the set of vertices is partitioned into a set of \textsc{And} vertices $V_\wedge$ and a set of \textsc{Or} vertices $V_\vee$. If the token is on an \textsc{Or} vertex, then player 0 moves the token, otherwise player 1 moves the token. Player 0 wins if the token finally reaches the target. The \textsc{And-Or Graph Reachability} problem is known to be \PTIME-complete~\cite{DBLP:journals/jcss/Immerman81}. 
The reduction in~\cite{erlebach2020game} unrolls the edge periodic graph into its configuration graph.
Considering the \CNR problem, we define for an edge periodic graph $\mathcal{G} = (V, E, \tau)$, the \emph{configuration graph} $\mathcal{C}_\mathcal{G} = ( V_{\mathcal{C}_\mathcal{G}}, E_{\mathcal{C}_\mathcal{G}})$ of $\mathcal{G}$ 
with node set $V \times V \times \{c, r\} \times [\LCM(L_\mathcal{G})]$.
Here, a node $(u_c, u_r, s, t)$ denotes in the temporal snapshot graph $\mathcal{G}(t)$, that the cop is on vertex $u_c$, the robber on vertex $u_r$, and the cop moves next if $s=c$, otherwise the robber moves next. 
A node $(u_c, u_r, s, t)$ is connected with a directed edge to some node $(u_c', u_r', s', (t+1) \mod \LCM(L_\mathcal{G}))$, if $s = c$, $s'=r$, $u_r' = u_r$, and $u_c' = u_c$ or $u_c$ and $u_c'$ are connected by an edge which is present in time step $t$. If $s=r$ the same holds for $c$ and $r$ interchanged. 
Hence, the configuration graph $\mathcal{C}_\mathcal{G}$ is of size $\mathcal{O}(|V|^2\cdot \LCM(L_\mathcal{G}))$.

The configuration graph $\mathcal{C}_\mathcal{G}$ is then turned into an \textsc{And-Or} graph by declaring nodes $(u_c, u_r, s, t)$ with $s=c$ as \textsc{Or} vertices and nodes with $s=r$ as \textsc{And} vertices. 
As $\LCM(L_\mathcal{G})$ can be of size exponential in $|L_\mathcal{G}|$, it follows from~\cite{erlebach2020game} that \CNR is contained in \EXPTIME. 
Note that the exponential blow-up comes from the very deterministic process of unrolling the edge periodic graph into a TVG with global periodicity $\LCM(L_\mathcal{G})$. 
Representing this TVG in a framework which only allows for global periodicity, such as the setting in~\cite{wehmuth2015unifying}, would require a representation of the TVG in size $\mathcal{O}(\LCM(L_\mathcal{G}))$ and allow for a polynomial time reduction to the \textsc{And-Or Graph Reachability} problem yielding membership in \PTIME for \CNR in this setting.

In the following, we show how use the structural properties of the configuration graph $\mathcal{C}_\mathcal{G}$ to solve the \textsc{And-Or Graph Reachability} for $\mathcal{C}_\mathcal{G}$ in polynomial space. 
We begin with giving an upper bound on the length of a shortest chase in $\mathcal{G}$.
\begin{lemma}
	\label{lem:max-chase}
	Let $\mathcal{G}=(V, E, \tau)$ be an edge periodic graph. If $\mathcal{G}$ is cop-winning, then the robber can be caught within $n^2 \cdot \LCM(L_\mathcal{G})$ rounds.
\end{lemma}
\begin{proof}
	Consider the configuration graph $\mathcal{C}_\mathcal{G} =( V_{\mathcal{C}_\mathcal{G}}, E_{\mathcal{C}_\mathcal{G}})$ of $\mathcal{G}$.
	Note that two configurations of a \CNR game on the same temporal snapshot graph with identical positions of the cop and robber in different time steps $t$ and $t'$ which differ by a multiple of $\LCM(L_\mathcal{G})$ are indistinguishable as configurations of the game. 
	Hence, the size of $\mathcal{C}_\mathcal{G}$ is bounded by $2 n^2 \cdot \LCM(L_\mathcal{G})$.
	
	Now let $\pi$ be a shortest sequence of configurations for the start vertices $v_c$ and $v_r$ such that $\pi$ ends with a cop-winning configuration. If $|\pi| > 2 n^2 \cdot \LCM(L_\mathcal{G})$, then $\pi$ contains two indistinguishable configurations $\pi_{i_1}$ and $\pi_{i_2}$. Clearly, the cop made no progress towards capturing the robber in the sequence $\pi_{i_1}, \dots, \pi_{i_2 -1}$. As $\pi_{i_1}$ and $\pi_{i_2}$ are indistinguishable, removing the sequence  $\pi_{i_1}, \dots, \pi_{i_2 -1}$ from $\pi$ yields a shorter sequence of configurations ending in a cop-winning configurations violating the assumed minimality of $\pi$. \lncsqed
\end{proof}
\begin{proof}[Proof of~\Cref{thm:PSPACE-Membership}]
	We will follow the ideas from~\cite{erlebach2020game} of reducing the \CNR problem to a variant of the \textsc{And-Or Graph Reachability} problem. 
	Therefore, we will need the notion of \emph{attractors} in an \textsc{And-Or} graph.
	Let $G=(V, E)$ be an \textsc{And-Or} graph with $V = V_\wedge \cup V_\vee$.
	Instead of considering a single target, we consider a set $T$ of targets and say that player 0 wins if the token finally reaches any state in $T$.
	Let $s\in V$ be the start vertex. Intuitively, the set of attractors of $G$ is the set of vertices from which player~0 can win the game. More formally, we inductively define the set of attractors $\att$ of $T$ as:
	\begin{align*}
		\att^0 &= T,\\
		\att^{i+1} &= \att_i \cup 
		\{v \in V_\wedge \mid \forall \{u, v\} \in E \colon u \in \att^i\}\ \cup\\
		& \qquad\quad\quad \ \{v \in V_\vee \mid \exists \{u, v\} \in E \colon u \in \att^i\},\\
		\att &= \bigcup_{i\geq 0} \att^i.
	\end{align*}
	
	Let $\mathcal{G} = (V, E, \tau)$ be the input edge periodic graph and let $\mathcal{C}_\mathcal{G} =( V_{\mathcal{C}_\mathcal{G}}, E_{\mathcal{C}_\mathcal{G}})$ be the configuration graph of $\mathcal{G}$. We will also identify $\mathcal{C}_\mathcal{G}$ as an \textsc{And-Or} graph by declaring nodes $(u_c, u_r, s, t)$ of $\mathcal{C}_\mathcal{G}$ with $s=c$ as \textsc{Or} vertices and nodes with $s=r$ as \textsc{And} vertices. 
	We define the set of nodes $ T = \{(u_c, u_r, s, t) \in V_{\mathcal{C}_\mathcal{G}} \mid u_c = u_r\}$ as the target set of the \textsc{And-Or} graph~$\mathcal{C}_\mathcal{G}$.
	We will now prove that the ideas from~\cite{erlebach2020game} of solving the \CNR game by checking if (i) there is some vertex $u_c$ such that for all vertices $u_r\in V$, the node $(u_c, u_r, c, 0)\in V_{\mathcal{C}_\mathcal{G}}$ is an attractor in the \textsc{And-Or} graph $\mathcal{C}_\mathcal{G}$, can be implemented in polynomial space. 
	Note that the set of attractors in $\mathcal{C}_\mathcal{G}$ corresponds to the set of configurations from which the cop has a winning strategy. 
	We use~\Cref{lem:max-chase} to unroll the configuration graph in order to obtain a leveled directed acyclic graph (DAG) which has width $n^2$ in each level and through which we can sweep level by level in order to verify property (i). 
	
	By~\Cref{lem:max-chase} we know that in order to verify property (i) it is sufficient to consider paths of length at most $2 n^2 \cdot \LCM(L_\mathcal{G})$ in $\mathcal{C}_\mathcal{G}$
	(the factor 2 is due to the alternation of players).  As the configuration graph is cyclic (due to modulo counting by $\LCM(L_\mathcal{G})$) we 
	unroll the graph $n^2$ times to allow for different time steps up to $n^2\cdot \LCM(L_\mathcal{G})$. The obtain DAG is big enough to contain any shortest chase starting in any time step $t \leq \LCM(L_\mathcal{G})$.
	The so obtained graph $\mathcal{C}'_\mathcal{G}$ consists of the node set $V_{\mathcal{C}'_\mathcal{G}} = V \times V \times \{c, r\} \times [n^2\cdot\LCM(L_\mathcal{G})]$ and the edge set $E_{\mathcal{C}'_\mathcal{G}}$ extending $E_{\mathcal{C}_\mathcal{G}}$ as $((u_c, u_r, s, t),$ $(u_c', u_r', s', t'))\in E_{\mathcal{C}'_\mathcal{G}}$ if and only if $t' \in \{t, t+1\}$, $((u_c, u_r, s, t\mod \LCM(L_\mathcal{G})),$ $(u_c', u_r', s', t'\mod \LCM(L_\mathcal{G})) \in E_{\mathcal{C}_\mathcal{G}}$ and $t, t' \leq n^2\cdot \LCM(L_\mathcal{G})-1$. 
	Verifying property (i) then corresponds to a reachability game in the \textsc{And-Or} graph associated with $\mathcal{C}'_\mathcal{G}$ with target set $T = \{(u_c, u_r, s, t) \in V_{\mathcal{C}'_\mathcal{G}} \mid u_c = u_r\}$ which can be solved using the notion of {attractors}.
%
	Note that in $\mathcal{C}'_\mathcal{G}$ only nodes with identical time steps or $t$ and $t+1$ are connected. Hence, in order to compute which nodes with time step $t$ belong to the set of attractors, we need to only know which nodes with time step $t+1$ are attractors. Since $\mathcal{C}'_\mathcal{G}$ is a DAG we can start in the level with $t = n^2\cdot\LCM(L_\mathcal{G})-1$ of  $\mathcal{C}'_\mathcal{G}$. 
	\begin{align*}
		\noalign{\noindent$\ \Atr^{n^2\cdot \LCM(L_\mathcal{G})-1}_r := \{(u_c, u_r, r, n^2\cdot \LCM(L_\mathcal{G})-1) \mid u_c = u_r\}$,}
		\Atr^{t}_r := &\{(u_c, u_r, r, t) \mid  \forall ((u_c, u_r, r, t), (u_c, u_r', c, t+1)) \in E_{\mathcal{C}'_\mathcal{G}}\colon (u_c, u_r', c, t+1)\\
		& \in \Atr_c^{t+1} \} 
		~\cup \{(u_c, u_r, r, t) \mid u_c = u_r\},
		\text{for } n^2\cdot \LCM(L_\mathcal{G})-2\geq t \geq 0,\\
		\Atr^{t}_c :=  &\{(u_c, u_r, c, t) \mid  \exists ((u_c, u_r, c, t), (u_c', u_r, r, t)) \in E_{\mathcal{C}'_\mathcal{G}}\colon (u_c', u_r, r, t) \in \Atr_r^{t} \} \\&
		~\cup \{(u_c, u_r, c, t) \mid u_c = u_r\},
		\text{for } n^2\cdot \LCM(L_\mathcal{G})-1 \geq t \geq 0.
	\end{align*}
	For each level $n^2\cdot \LCM(L_\mathcal{G})-1 \geq t \geq 0$ we only need to keep the last level ($t+1$ if existent)\footnote{Note that we can easily compute the snapshot $\mathcal{G}(n^2\cdot \LCM(L_\mathcal{G})) = \mathcal{G}(0)$ by drawing all edges with $\modid{\tau(e)}{0} = 1$; and from $\mathcal{G}(t)$ for some time step $t$, the snapshot $\mathcal{G}(t-1)$ by shifting the pointer in each $\tau(e)$ one step to the left. Therefore, we can compute from each level $t+1$ of $\mathcal{C}'_\mathcal{G}$ the level $t$ in polynomial time and space.}	
	of $\mathcal{C}'_\mathcal{G}$ in memory in order to compute the sets $\Atr^{t}_c$ and $\Atr^{t}_r$ of nodes $(u_c, u_r, s, t)$ in $\mathcal{C'}_\mathcal{G}$ from which the cop has a winning strategy where $s$ equals the sub-script. 
	Note that $\bigcup_{0 \leq t \leq n^2\cdot \LCM(L_\mathcal{G})-1} \Atr^{t}_c \cup \Atr^{t}_r = \Atr$. In order to verify property (i) we only need to keep the current and latest sets $\Atr^{t}_c$, $\Atr^{t}_r$, $\Atr^{t+1}_c$, $\Atr^{t+1}_r$ in memory yielding a polynomial space algorithm. 
\lncsqed
\end{proof}
\section{Discussion}
\label{sec:discussion}
While we improved the currently known upper bound for the \CNR problem on edge periodic graphs from \EXPTIME to \PSPACE and improved the lower bounds, of being \NP-hard, to include also the very restrictive classes of directed and undirected edge periodic cycles, a gap in the complexity of \CNR remains. 
It is worth noticing that on one side the chosen representation of edge periodic graphs is quite  compact, as a natural proof for a cop-winning strategy might be of exponential length in the input size, since the periodicity of the whole graph is the least common multiple of the periodicity of each edge, which prevents the use of a simple guess \& check approach for \NP-membership. 
On the other side, the chosen representation is still exponentially larger than the representation by on-line programs used in~\cite{sutner1988motion} where \PSPACE-completeness for the reachability problem on a related but different class of periodic TVGs was obtained.

If we consider
 \emph{directed} edge periodic \emph{cycles}, then determining whether the given cycle is cop-winning boils down to deterministically simulating the chase starting from a (guessed) cop vertex and time step, as the optimal strategies for the cop and robber are both to keep running whenever possible (without bumping into the cop).
For the robber the optimal start vertex is directly behind the cop.
Since $\LCM(L_\mathcal{G})$ can be exponentially large in the size of~$\mathcal{G}$ the only known upper bound on the number of steps in the simulation of the chase starting in time step~$t$ is exponential in the size of $\mathcal{G}$, while the chase does not reveal any complexity.
%
 The simulation could even be performed by a log-space Turing-Machine being equipped with a clock which allows for modulo queries of logarithmic size. To better understand the precise complexity of \CNR on directed edge periodic cycles, the theoretical analysis of potential families of cycles with shortest cop-winning strategies of exponential size would be of great interest and might indicate the necessity for a new complexity class consisting of simple simulation problems with exponential duration time.

\bibliographystyle{plain}
\bibliography{bib-slim}

\clearpage
\appendix
\section{Details on~\Cref{lem:famCopWin}}
We explicitly give the edge periodic cycles for $k=2$, $k=3$, and $k=4$ in the proof of~\Cref{lem:famCopWin}. 
For $k=2$ and $k=3$ the chase of the cop will be shorter than described in Table~\ref{tab:4k-1} and for $k\geq4$ the chase will be exactly as described in general in Table~\ref{tab:4k-1}.
The edge periodic cycle for $k=2$ is depicted in Figure~\ref{fig:k2} and the chase is described in Table~\ref{tab:k2}. 
For $k=3$ the edge periodic cycle is depicted in Figure~\ref{fig:k3} and the chase is described in Table~\ref{tab:k3}.
Finally, for $k=4$, the edge periodic cycle is depicted in Figure~\ref{fig:k4} and the explicit chase is described in Table~\ref{tab:k4}. Note that Table~\ref{tab:k4} is identical to Table~\ref{tab:4k-1} if we set $k=4$ in Table~\ref{tab:4k-1}.

%
%
\begin{figure}[ht]
	\centering
	\begin{tikzpicture}[node distance=15pt,,yscale=0.9, state/.style={circle, draw,minimum size=3pt, initial text = }]
		\node[state,fill=red,minimum size=3pt] (0) at (5, 2) {$0$};
		\node[state] (1) at (7, 2) {$1$};
		\node[state] (2) at (8, 0.5) {$2$};
		\node[state] (3) at (6, 0) {$3$};
		\node[state] (4) at (4, 0) {$4$};
		\node[state] (5) at (2, 0.5) {$5$};
		\node[state] (6) at (3, 2) {$6$};
		\path (0) edge node[above,fill=gray!20] {$10$} (1)
		(1) edge (2) 
		(2) edge (3)
		(3) edge node[above,fill=gray!20] {$01$} (4)
		(4) edge node[above,fill=gray!20] {$10$} (5)
		(5) edge (6)
		(6) edge (0)
		; 
	\end{tikzpicture}
	\caption{Edge periodic cycle for the case $k=2$ in~\Cref{lem:famCopWin} with $4\cdot k-1 = 7$ vertices and $\LCM(L_{\mathcal{G}})=2$ with a cop winning strategy from the start vertex marked in red. Edges without edge label are 1-edges; for all other edges, $\tau(e)$ is explicitly noted (with gray background).}
	\label{fig:k2}
\end{figure}
\begin{table}[ht]
	\centering
	\begin{tabular}{r || r | r}
		time step & pos.~cop & pos.~robber\\
		\hline
		$s$ & 0 & 3\\
		0 & 1 & 3\\
		1 & 2 & 4\\
		2 & 3 & 5\\
		3 & 4 & 6\\
		4 & 5 & 0\\
		5 & 6 & 0\\
		6 & 0 & $\skull$\\
		\multicolumn{3}{l}{\ } \\
		\multicolumn{3}{l}{\ } \\
	\end{tabular}
	\hfil
	\begin{tabular}{r || r | r}
		time step & pos.~cop & pos.~robber\\
		\hline
		$s$ & 0 & 4\\
		0 & 6 & 4\\
		1 & 5 & 3\\
		2 & 4 & 2\\
		3 & 3 & 1\\
		4 & 2 & 0\\
		5 & 1 & 6\\
		6 & 0 & 5\\
		7 & 6 & 5\\
		8 & 5 &	$\skull$
	\end{tabular}
	\caption{Time steps with corresponding positions of cop and robber in the edge periodic cycle depicted in Figure~\ref{fig:k2}. All positions are \emph{after} moving in this time step. The time step $s$ denotes the start configuration. Recall that the cop moves first.}
	\label{tab:k2}
\end{table}
\clearpage
%
%
%
\begin{figure}[ht]
	\centering
	\begin{tikzpicture}[node distance=15pt,,yscale=0.9, state/.style={circle, draw,minimum size=20pt, initial text = }]
		\node[state,fill=red] (0) at (5, 2) {$0$};
		\node[state] (1) at (7, 2) {$1$};
		\node[state] (2) at (9, 2) {$2$};
		\node[state] (3) at (10, 1) {$3$};
		\node[state] (4) at (9, 0) {$4$};
		\node[state] (5) at (6, 0) {$5$};
		\node[state] (6) at (4, 0) {$6$};
		\node[state] (7) at (2, 0) {$7$};
		\node[state] (8) at (0, 0) {$8$};
		\node[state] (9) at (0, 2) {$9$};
		\node[state] (10) at (2, 2) {$10$};

		\path (0) edge node[above,fill=gray!20] {$100$} (1)
		(1) edge (2) 
		(2) edge (3)
		(3) edge (4)
		(4) edge (5)
		(5) edge node[above,fill=gray!20] {$001$} (6)
		(6) edge (7)
		(7) edge node[above,fill=gray!20] {$100$} (8)
		(8) edge (9)
		(9) edge (10)
		(10) edge (0)
		; 
	\end{tikzpicture}
	\caption{Edge periodic cycle for the case $k=3$ in~\Cref{lem:famCopWin} with $4\cdot k-1 = 11$ vertices and $\LCM(L_{\mathcal{G}})=3$ with a cop winning strategy from the start vertex marked in red. Edges without edge label are 1-edges; for all other edges, $\tau(e)$ is explicitly noted (with gray background).}
	\label{fig:k3}
\end{figure}
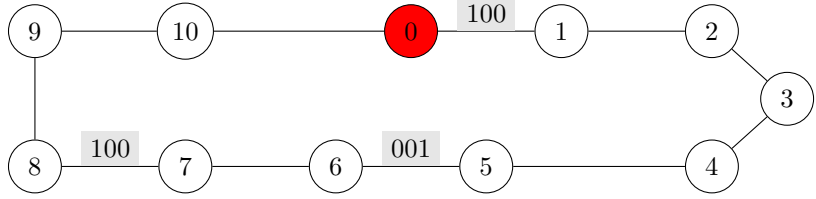
\begin{table}[ht]
	\centering
	\begin{tabular}{r || r | r}
		time step & pos.~cop & pos.~robber\\
		\hline
		$s$ & 0 & 5\\
		0 & 1 & 5\\
		1 & 2 & 5\\
		2 & 3 & 6\\
		3 & 4 & 7\\
		4 & 5 & 7\\
		5 & 6 & 7\\
		6 & 7 & $\skull$\\
		\multicolumn{3}{l}{\ } \\
		\multicolumn{3}{l}{\ } \\
		\multicolumn{3}{l}{\ } \\
	\end{tabular}
	\hfil
	\begin{tabular}{r || r | r}
		time step & pos.~cop & pos.~robber\\
		\hline
		$s$ & 0 & 6\\
		0 & 10 & 6\\
		1 & 9 & 6\\
		2 & 8 & 5\\
		3 & 7 & 4\\
		4 & 6 & 3\\
		5 & 5 & 2\\
		6 & 4 & 1\\
		7 & 3 & 1\\
		8 & 2 & 1\\
		9 & 1 & $\skull$
	\end{tabular}
	\caption{Time steps with corresponding positions of cop and robber in the edge periodic cycle depicted in Figure~\ref{fig:k3}. All positions are \emph{after} moving in this time step. The time step $s$ denotes the start configuration. Recall that the cop moves first.}
	\label{tab:k3}
\end{table}
%
%
%
%
\begin{figure}[ht]
	\centering
	\begin{tikzpicture}[node distance=15pt,,yscale=0.9, state/.style={circle, draw,minimum size=20pt, initial text = }]
		\node[state,fill=red] (0) at (5, 2) {$0$};
		\node[state] (1) at (7, 2) {$1$};
		\node[state] (2) at (8.5, 2) {$2$};
		\node[state] (3) at (10, 2) {$3$};
		\node[state] (4) at (11, 1) {$4$};
		\node[state] (5) at (10, 0) {$5$};
		\node[state] (6) at (8.5, 0) {$6$};
		\node[state] (7) at (6, 0) {$7$};
		\node[state] (8) at (4, 0) {$8$};
		\node[state] (9) at (2.5, 0) {$9$};
		\node[state] (10) at (1, 0) {$10$};
		\node[state] (11) at (-1, 1) {$11$};
		\node[state] (12) at (0.5, 2) {$12$};
		\node[state] (13) at (2, 2) {$13$};
		\node[state] (14) at (3.5, 2) {$14$};
		
		\path (0) edge node[above,fill=gray!20] {$1000$} (1)
		(1) edge (2) 
		(2) edge (3)
		(3) edge (4)
		(4) edge (5)
		(5) edge (6)
		(6) edge (7)
		(7) edge node[above,fill=gray!20] {$0001$} (8)
		(8) edge (9)
		(9) edge (10)
		(10) edge node[above,fill=gray!20] {$1000$} (11)
		(11) edge (12)
		(12) edge (13)
		(13) edge (14)
		(14) edge (0)
		; 
	\end{tikzpicture}
	\caption{Edge periodic cycle for the case $k=4$ in~\Cref{lem:famCopWin} with $4\cdot k-1 = 15$ vertices and $\LCM(L_{\mathcal{G}})=4$ with a cop winning strategy from the start vertex marked in red. Edges without edge label are constant 1-edges; for all other edges, $\tau(e)$ is explicitly noted (with gray background).}
	\label{fig:k4}
\end{figure}
\begin{table}[ht]
	\centering
	\begin{tabular}{r || r | r}
		time step & pos.~cop & pos.~robber\\
		\hline
		$s$ & 0 & 7\\
		0 & 1 & 7\\
		1 & 2 & 7\\
		2 & 3 & 7\\
		3 & 4 & 8\\
		4 & 5 & 9\\
		5 & 6 & 10\\
		6 & 7 & 10\\
		7 & 8 & 10\\
		8 & 9 & 11\\
		9 & 10 & 12\\
		10 & 10 & 13\\
		11 & 10 & 14\\
		12 & 11 & 0\\
		13 & 12 & 0\\
		14 & 13 & 0\\
		15 & 14 & 0\\
		16 & 0 & $\skull$\\
		\multicolumn{3}{l}{\ } \\
		\multicolumn{3}{l}{\ } \\
		\multicolumn{3}{l}{\ } \\
		\multicolumn{3}{l}{\ } \\
	\end{tabular}
	\hfil
	\begin{tabular}{r || r | r}
		time step & pos.~cop & pos.~robber\\
		\hline
		$s$ & 0 & 8\\
		0 & 14 & 8\\
		1 & 13 & 8\\
		2 & 12 & 8\\
		3 & 11 & 7\\
		4 & 10 & 6\\
		5 & 9 & 5\\
		6 & 8 & 4\\
		7 & 7 & 3\\
		8 & 6 & 2\\
		9 & 5 & 1\\
		10 & 4 & 1\\
		11 & 3 & 1\\
		12 & 2 & 0\\
		13 & 1 & 14\\
		14 & 1 & 13\\
		15 & 1 & 12\\
		16 & 0 & 11\\
		17 & 14 & 11\\
		18 & 13 & 11\\
		19 & 12 & 11\\
		20 & 11 & $\skull$
	\end{tabular}
	\caption{Time steps with corresponding positions of cop and robber in the edge periodic cycle depicted in Figure~\ref{fig:k4}. Note that the position of the cop and robber are as described in Table~\ref{tab:4k-1} for the general case of $k\geq 4$. All positions are \emph{after} moving in this time step. The time step $s$ denotes the start configuration. Recall that the cop moves first.}
	\label{tab:k4}
\end{table}

 \clearpage
\section{Details on~\Cref{lem:famCopWin2i}}
We explicitly give the edge periodic cycle for $k=4$ in the proof of~\Cref{lem:famCopWin2i}. 
The edge periodic cycle is depicted in Figure~\ref{fig:k2i} and the explicit chase is described in Table~\ref{tab:k2i}.
 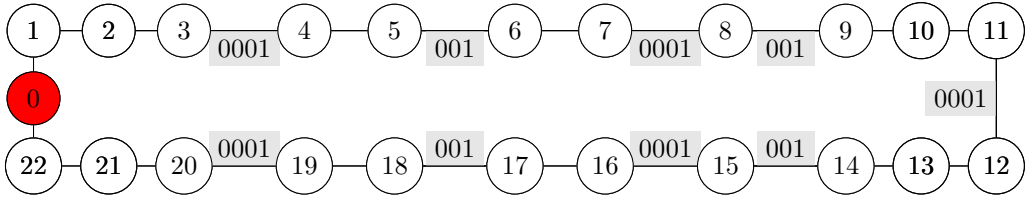
\begin{figure}[ht]
	\centering
	\begin{tikzpicture}[node distance=15pt,,yscale=0.9, state/.style={circle, draw,minimum size=20pt, initial text = }]
		\node[state,fill=red] (0) at (0, 0) {$0$};

		\node[state] (1) at (0, 1) {$1$};
		\node[state] (22) at (0, -1) {$22$};
		\node[state] (11) at (12.8, 1) {$11$};
		\node[state] (12) at (12.8, -1) {$12$};
		
		\node[state] (2) at (1, 1) {$2$};
		\node[state] (21) at (1, -1) {$21$};
		\node[state] (10) at (11.8, 1) {$10$};
		\node[state] (13) at (11.8, -1) {$13$};
		\foreach \i in {3,...,4}
		{
        \pgfmathtruncatemacro{\x}{23-\i};
		\node[state, fill=white] (\i) at (\i * 1.6 - 2.8, 1) {$\i$};
		\node[state, fill=white] (\x) at (\i * 1.6 - 2.8, -1) {$\x$};
		}
		
		\foreach \i in {5,...,6}
		{
        \pgfmathtruncatemacro{\x}{23-\i};
		\node[state, fill=white] (\i) at (\i * 1.6 - 3.2, 1) {$\i$};
		\node[state, fill=white] (\x) at (\i * 1.6 - 3.2, -1) {$\x$};
		}
		
		\foreach \i in {7,...,9}
		{
        \pgfmathtruncatemacro{\x}{23-\i};
		\node[state, fill=white] (\i) at (\i * 1.6 - 3.6, 1) {$\i$};
		\node[state, fill=white] (\x) at (\i * 1.6 - 3.6, -1) {$\x$};
		}
		
		\foreach \i in {1,...,21}
		{
        \pgfmathtruncatemacro{\x}{\i+1};
        \draw (\i) to (\x);
		}

		\foreach \i in {2,3}
		{
        \pgfmathtruncatemacro{\x}{\i*3};
        \pgfmathtruncatemacro{\y}{\x-1};
        \draw (\y) edge node[below,fill=gray!20] {$001$} (\x);
		}
		
		\foreach \i in {5,6}
		{
        \pgfmathtruncatemacro{\x}{\i*3};
        \pgfmathtruncatemacro{\y}{\x-1};
        \draw (\y) edge node[above,fill=gray!20] {$001$} (\x);
		}
		
        \draw (3) edge node[below,fill=gray!20] {$0001$} (4);
        \draw (7) edge node[below,fill=gray!20] {$0001$} (8);
        \draw (11) edge node[left,fill=gray!20] {$0001$} (12);
        \draw (15) edge node[above,fill=gray!20] {$0001$} (16);
        \draw (19) edge node[above,fill=gray!20] {$0001$} (20);
		\draw (0) to (1);
		\draw (0) to (22);

\node[state,fill=red] (0) at (0, 0) {$0$};

		\node[state] (1) at (0, 1) {$1$};
		\node[state] (22) at (0, -1) {$22$};
		\node[state] (11) at (12.8, 1) {$11$};
		\node[state] (12) at (12.8, -1) {$12$};
		
		\node[state] (2) at (1, 1) {$2$};
		\node[state] (21) at (1, -1) {$21$};
		\node[state] (10) at (11.8, 1) {$10$};
		\node[state] (13) at (11.8, -1) {$13$};
		\foreach \i in {3,...,4}
		{
        \pgfmathtruncatemacro{\x}{23-\i};
		\node[state, fill=white] (\i) at (\i * 1.6 - 2.8, 1) {$\i$};
		\node[state, fill=white] (\x) at (\i * 1.6 - 2.8, -1) {$\x$};
		}
		
		\foreach \i in {5,...,6}
		{
        \pgfmathtruncatemacro{\x}{23-\i};
		\node[state, fill=white] (\i) at (\i * 1.6 - 3.2, 1) {$\i$};
		\node[state, fill=white] (\x) at (\i * 1.6 - 3.2, -1) {$\x$};
		}
		
		\foreach \i in {7,...,9}
		{
        \pgfmathtruncatemacro{\x}{23-\i};
		\node[state, fill=white] (\i) at (\i * 1.6 - 3.6, 1) {$\i$};
		\node[state, fill=white] (\x) at (\i * 1.6 - 3.6, -1) {$\x$};
		}

		\foreach \i in {1,...,21}
		{
        \pgfmathtruncatemacro{\x}{\i+1};
        \draw (\i) to (\x);
		}
		
	\end{tikzpicture}
	\caption{Cycle with $23 = 6k-1$ vertices and $\LCM(L_{\mathcal{G}})=12=3k$ with a cop winning strategy from the start vertex~$0$ where~$k=4$. 
	Edges without an explicit label are 1-edges.}
	\label{fig:k2i}
\end{figure}
\begin{table}[ht]
	\centering
	\begin{tabular}{r || r | r}
		time step & pos.~cop & pos.~robber\\
		\hline
		$s$ & 0 & 11\\
		0 & 1 & 11\\
		1 & 2 & 11\\
		2 & 3 & 11\\
		3 & 4 & 12\\
		4 & 5 & 13\\
		5 & 6 & 14\\
		6 & 7 & 14\\
		7 & 8 & 14\\
		8 & 9 & 15\\
		9 & 10 & 15\\
		10 & 11 & 15\\
		11 & 12 & 16\\
		12 & 13 & 17\\
		13 & 14 & 17\\
		14 & 15 & 18\\
		15 & 16 & 19\\
		16 & 17 & 19\\
		17 & 18 & 19\\
		18 & 19 & $\skull$
	\end{tabular}
	\hfil
	\begin{tabular}{r || r | r}
		time step & pos.~cop & pos.~robber\\
		\hline
		$s$ & 0 & 12\\
		0 & 22 & 12\\
		1 & 21 & 12\\
		2 & 20 & 12\\
		3 & 19 & 11\\
		4 & 18 & 10\\
		5 & 17 & 9\\
		6 & 16 & 9\\
		7 & 15 & 9\\
		8 & 14 & 8\\
		9 & 13 & 8\\
		10 & 12 & 8\\
		11 & 11 & 7\\
		12 & 10 & 6\\
		13 & 9 & 6\\
		14 & 8 & 5\\
		15 & 7 & 4\\
		16 & 6 & 4\\
		17 & 5 & 4\\
		18 & 4 & $\skull$
	\end{tabular}
	\caption{Time steps with corresponding positions of cop and robber in the edge periodic cycle depicted in Figure~\ref{fig:k2i}.  All positions are \emph{after} moving in this time step. The time step $s$ denotes the start configuration. Recall that the cop moves first.}
	\label{tab:k2i}
\end{table}

\end{document}